\documentclass[11pt,a4paper]{amsart}
\usepackage[dvips,colorlinks=true,linkcolor=blue,urlcolor=red]{hyperref}
\usepackage[dvips]{graphicx}
\usepackage{subfigure,graphicx,caption}
\newcommand{\R}{\mathbb{R}}

\newcommand{\N}{\mathbb{N}}

\newcommand{\Z}{\mathbb{Z}}

\newcommand{\Q}{\mathbb{Q}}

\theoremstyle{plain}
\newtheorem{Th}{Theorem}
\newtheorem{Le}{Lemma}

\theoremstyle{definition}
\newtheorem{Rem}{Remark}
\newtheorem{Def}{Definition}
\title[Landau Hamiltonian perturbed by a generic periodic
potential]{Absolute continuity of the spectrum of a Landau Hamiltonian
  perturbed by a generic periodic potential}
\author{Fr{\'e}d{\'e}ric Klopp}
\address[Fr{\'e}d{\'e}ric Klopp]{LAGA, U.M.R. 7539 C.N.R.S, Institut Galil{\'e}e,
  Universit{\'e} Paris-Nord, 99 Avenue J.-B. Cl{\'e}ment, F-93430
  Villetaneuse, France\\ et\\ Institut Universitaire
  de France}
\email{\href{mailto:klopp@math.univ-paris13.fr}{klopp@math.univ-paris13.fr}}
\thanks{Part of this work was done during the conference ``Spectral
  analysis of differential operators'' held at the MFO, Oberwolfach
  (29/11-03/12/2004); it is a pleasure to thank T. Weidl and A.
  Sobolev, the organizers of the conference, for their invitation to
  participate to the meeting as well as D. Elton for stimulating
  discussions.\\ It is also a pleasure to thank the Institute of
  Mathematics of Hanoi, where this work was completed, for its kind
  hospitality.}
\begin{document}
%
\begin{abstract}
  Consider $\Gamma$, a non-degenerate lattice in $\R^2$ and a constant
  magnetic field $B$ with a flux though a cell of $\Gamma$ that is a
  rational multiple of $2\pi$.  We prove that for a generic
  $\Gamma$-periodic potential $V$, the spectrum of the Landau
  Hamiltonian with magnetic field $B$ and periodic potential $V$ is
  purely absolutely continuous.
  \vskip.5cm\noindent \textsc{R{\'e}sum{\'e}.}
  On consid{\`e}re $\Gamma$, un r{\'e}seau non-d{\'e}g{\'e}n{\'e}r{\'e} dans $\R^2$ et un
  champ magn{\'e}tique constant $B$ dont le flux {\`a} travers une cellule du
  r{\'e}seau est un multiple rationnel de $2\pi$. On d{\'e}montre que, pour un
  potentiel $\Gamma$-p{\'e}riodique $V$ continu g{\'e}n{\'e}rique, le spectre du
  hamiltonien de Landau de champ magn{\'e}tique constant $B$ perturb{\'e} par
  le potentiel p{\'e}riodique $V$ est purement absolument continu.
\end{abstract}
\maketitle
Written in the Coulomb gauge, on $L^2(\R^2)$, the Landau Hamiltonian is
defined by
\begin{equation}
  \label{eq:1}
  H= (-i\nabla-A)^2,\quad\text{ where }\quad A(x_1,x_2)=
  \frac{B}{2}(-x_2,x_1),
\end{equation}
Let $\Gamma=\oplus_{i=1}^2\Z e_i$ be a non-degenerate lattice such
that
\begin{equation}
  \label{eq:2}
  B\, e_1\wedge e_2\in2\pi\Q.
\end{equation}
Define the set of real valued, continuous, $\Gamma$-periodic functions
\begin{equation}
  \label{eq:3}
  C_\Gamma=\{V\in C(\R^2,\R);\ \forall x\in\R^2,\
  \forall\gamma\in\Gamma,\ V(x+\gamma)=V(x)\}.
\end{equation}
The space $C_\Gamma$ is endowed with the uniform topology, the
associated norm being denoted by $\|\cdot\|$.\\
Our main result is
\begin{Th}
  \label{thr:1}
  There exists a $G_\delta$-dense subset of $C_\Gamma$ such that,
  for $V$ in this set, the spectrum of $H(V):=H+V$ is purely
  absolutely continuous.
\end{Th}
\noindent The absence of singular continuous spectrum can be obtained
from the sole analytic direct integral representation of $H(V)$ that we
use below (\cite{MR2120804}).\\
Our result is optimal in the sense that there are examples of periodic
$V$ for which the spectrum of $H$ contains eigenvalues such as $V$
constant. Of course, it is a natural question to wonder whether the
constant potential is the only periodic one for which the spectrum
exhibits eigenvalues.
\vskip.2cm The proof of Theorem~\ref{thr:1} consists in several
steps. We first reduce the problem via magnetic Floquet
theory. Therefore, we introduce the magnetic
translations~\cite{Sj:91}. For the two-dimensional, constant,
transverse magnetic field problem, they are defined as follows. For
any field strength $B\in\R$, any vector $\alpha\in\R^2$, and $f\in
C_0^\infty(\R^2)$, we define the magnetic translation by $\alpha$ to
be
\begin{equation}
  \label{eq:4}
  U_\alpha^Bf(x):=e^{\frac{iB}{2}x\wedge\alpha}
  f(x+\alpha)=e^{\frac{iB}{2}(x_1\alpha_2-x_2\alpha_1)}
  f(x+\alpha).
\end{equation}
For $(\alpha,\beta)\in(\R^2)^2$, we have the commutation relations
\begin{equation}
  \label{eq:5}
  U_\alpha^B U_\beta^B=e^{iB\,\alpha\wedge\beta}\,U_\beta^BU_\alpha^B .
\end{equation}
In a standard way, the family $\{U_\alpha^B;\;\alpha\in\R^2\}$ extends
to a projective unitary representation of $\R^2$ on $L^2(\R^2)$.  We
note that
\begin{equation}
  \label{eq:6}
  [U_\alpha^B,H]=0\quad\text{and}\quad[U_\alpha^B,V]=0.
\end{equation}
Let $(e_1,e_2)$ be a ``fundamental basis'' of the lattice $\Gamma$
i.e. $\Gamma=\oplus_{i=1}^2\Z e_i$. For $j\in\{1,2\}$, we define the
unitary $U_j^B:=U_{e_j}^B$ by~\eqref{eq:4}. By
assumption~\eqref{eq:2}, one has
\begin{equation}
  \label{eq:9}
  B e_1\wedge e_2=2\pi p/q,\quad\text{for}\quad(p,q)\in\Z\times\N,\
  p\wedge q=1.
\end{equation}
It follows from~\eqref{eq:2} and~\eqref{eq:5} that the unitary
operators $\{ (U_1^B)^q, U_2^B \}$ satisfy the commutation relation
\begin{equation*}
  (U_1^B)^qU_2^B=e^{iqBe_1\wedge e_2}U_2^B(U_1^B)^q
  =e^{i2\pi p}U_2^B (U_1^B)^q=U_2^B(U_1^B)^q,
\end{equation*}
so the pair generates an Abelian group.\\
One checks that
\begin{equation}
  \label{eq:7}
  [(U_1^B)^q,H(V)]=0=[U_2^B,H(V)] 
\end{equation}
Consider $\Gamma'$, the sublattice of $\Gamma$ defined by $\Gamma'=q\Z
e_1\oplus\Z e_2$. Its dual lattice $(\Gamma')^*$ is given by
\begin{equation*}
  (\Gamma')^*=\{\gamma^*\in\R^2;\ \forall\gamma'\in\Gamma',\
  \gamma^*\cdot \gamma'\in2\pi\Z\}.
\end{equation*}
For any $\gamma'=q\gamma'_1e_1+\gamma'_2e_2\in\Gamma'$, define the
phase $\Theta_q(\gamma')$ by
\begin{equation}
  \label{eq:10}
  \Theta_q(\gamma')=e^{iBe_1\wedge e_2q\gamma'_1\gamma'_2/2}
  =e^{i\pi p\gamma'_1\gamma'_2}\in\{-1,+1\}.
\end{equation}
This allows us to define a unitary representation of the sublattice
$\Gamma'$ by
\begin{equation}
  \label{eq:8}
  W_{q,\gamma'}^B=\Theta_q(\gamma')U_{\gamma'}^B.
\end{equation}
It is easy to check that
\begin{equation*}
  W_{q,\gamma}^BW_{q,\gamma'}^B=W_{q,\gamma+\gamma'}^B
  ,\quad\forall(\gamma,\gamma')\in(\Gamma')^2.
\end{equation*}
We define the transformation $T^B$ on smooth functions by
\begin{equation*}
  (T^Bf)(x,\theta)=\sum_{\gamma'\in\Gamma'}\; e^{i\theta\cdot (x+\gamma)}
  (W_{q,\gamma'}^Bf)(x),\quad \theta\in(\R^2)^*/(\Gamma')^*.
\end{equation*}
Again, a simple calculation shows that
\begin{equation*}
  (W_{q,\gamma'}^BT^Bf)(x,\theta)=(T^Bf)(x,\theta).
\end{equation*}
We define a function space ${\mathcal H}_{B,p}$ by
\begin{equation*}
  {\mathcal H}_{B,p}=\{v\in L_{loc}^2(\R^2)\;|\;
  W_{q,\gamma'}^Bv=v;\ \forall\gamma'\in\Gamma'\}.
\end{equation*}
It then follows that $T^B$ extends to a unitary map
\begin{equation*}
  T^B:\ L^2(\R^2)\rightarrow L^2((\R^2)^*/(\Gamma')^*),{\mathcal
    H}_{B,p}).
\end{equation*}
Given this structure, it is clear that the Hamiltonian $H$ admits a
direct integral decomposition (see e.g.~\cite{MR58:12429c}) over
$(\R^2)^*/(\Gamma')^*$, so that
\begin{equation*}
  T^BH(V)(T^B)^*=\int^{\oplus}_{(\R^2)^*/(\Gamma')^*}\,
  H(\theta,V)\,d\theta.
\end{equation*}
The operator $H(\theta,V)$ is self-adjoint on the Sobolev space
${\mathcal H}_{B,p}^2$, the local Sobolev space of order two of
functions in ${\mathcal H}_{B,p}$ and one computes
\begin{equation}
  \label{eq:11}
  H(\theta,V)=(i\nabla+A-\theta)^2+V.
\end{equation}
This operator has a compact resolvent. Consequently, the spectrum is
discrete and consists of eigenvalues of finite multiplicity, say,
$(E_j(V,\theta))_{j\in\{1,2,\ldots\}}$ labeled in increasing order and
repeated according to multiplicity. For $n\geq1$, the function
$(\theta,V)\in(\R^2)^*/(\Gamma')^*\times C_\Gamma\mapsto E_n(V,\theta)$ is
locally uniformly Lipschitz continuous; this follows from the
variational principle (see e.g.~\cite{MR58:12429c}) and the fact that
$(i\nabla-A)$ is $H$-bounded with relative bound 0. We endow the space
$(\R^2)^*/(\Gamma')^*\times C_\Gamma$ with the norm
$\|(\theta,V)\|=|\theta|+\|V\|$.
\vskip.1cm\noindent It is well known
(see~\cite{MR58:12429c,MR1703339}) that Theorem~\ref{thr:1} is a
corollary of
\begin{Th}
  \label{thr:2}
  There exists a $G_\delta$-dense subset of $C_\Gamma$ such that, for
  $V$ in this set, none of the functions $\theta\mapsto
  E_n(V,\theta)$, $n\geq1$, is constant.
\end{Th}

\noindent Pick $\theta_0\in(\R^2)^*/(\Gamma')^*$ and $V\in
C_\Gamma$. Let $n\geq1$.
\begin{Def}
  \label{def:1}
  $E_n(\theta_0,V_0)$ is an analytically degenerate eigenvalue of
  $H(\theta_0,V_0)$ if and only if there exists $\delta>0$ and an
  orthonormal system of $p$ functions, say $(\theta,V)\mapsto
  \varphi_j(\cdot,\theta,V)$, $j\in\{1,\cdots,p\}$, defined and real
  analytic on
  $U_{\theta_0,V_0}:=\{||(\theta,V)-(\theta_0,V_0)\|<\delta\}$ valued
  in ${\mathcal H}_{B,p}^2$ such that, for all $(\theta,V)\in
  U_{\theta_0,V_0}$,
  \begin{itemize}
  \item the functions $(\varphi_j(\cdot,\theta,V))_{1\leq j\leq p}$
    span the kernel of $H(\theta,V)-E_n(\theta,V)$,
  \item one has
    \begin{equation*}
      H(\theta,V)\varphi_j(\theta,V)=E_n(\theta,V)\varphi_j(\theta,V)\quad
      \text{for}\quad 1\leq j\leq p.
    \end{equation*}
  \end{itemize}
\end{Def}
\begin{Rem}
  \label{rem:2}
  As one can see from the proof of Lemma~\ref{le:2}, to say that
  $E_n(\theta,V)$ is analytically degenerate near $(E_0,V_0)$ is
  equivalent to say that the multiplicity of $E_n(\theta,V)$ is
  constant in some neighborhood of $(E_0,V_0)$.
\end{Rem}
\noindent Theorem~\ref{thr:2} is a consequence of the following two
lemmas
\begin{Le}
  \label{le:1}
  Pick $\theta_0\in(\R^2)^*/(\Gamma')^*$ and $V_0\in C_\Gamma$ such
  that $V_0$ is not a constant. Assume that $E_n(\theta_0,V_0)$ is an
  analytically degenerate eigenvalue of $H(\theta_0,V_0)$.  Then, for
  any $\varepsilon>0$, there exists $V\in\{\|V-V_0\|< \varepsilon\}$
  such that $\theta\mapsto E_n(\theta,V)$ is not constant.
\end{Le}
\noindent and
%
\begin{Le}
  \label{le:2}
  Pick $\theta_0\in(\R^2)^*/(\Gamma')^*$ and $V_0\in C_\Gamma$. Fix
  $n\geq1$. Then, for any $\varepsilon>0$, there exists
  $(\theta_\varepsilon,V_\varepsilon)\in\{\|(\theta,V)-(\theta_0,V_0)\|
  <\varepsilon\}$ and $\delta>0$ such that $E_n(\theta,V)$ is an
  analytically degenerate eigenvalue of $H(\theta,V)$ for
  $(\theta,V)\in\{\|(\theta,V)-(\theta_\varepsilon,V_\varepsilon)\|<\delta\}$.
\end{Le}
\begin{Rem}
  \label{rem:1}
  In general, in Lemma~\ref{le:2}, the multiplicity of the eigenvalue
  is one.
\end{Rem}
\noindent How to complete the proof of Theorem~\ref{thr:2} using
Lemmas~\ref{le:1} and~\ref{le:2} is straightforward. For any $n\geq
1$, the set of $V$ in $C_\Gamma$ such that $\theta\mapsto
E_n(\theta,V)$ is not constant is open (as the Floquet eigenvalues are
locally uniformly Lipschitz continuous in $(\theta;V)$). In view of
Lemma~\ref{le:1} and~\ref{le:2}, for any $n\geq1$, the set of $V$ in
$C_\Gamma$ such that $\theta\mapsto E_n(\theta,V)$ is not constant is
dense.  Hence, the set of $V$ where none of $(\theta\mapsto
E_n(\theta,V))_{n\geq1}$ is constant is a countable intersection of
dense open sets i.e. a $G_\delta$-dense set. This completes the proof
of Theorem~\ref{thr:2}.
\begin{proof}[The proof of Lemma~\ref{le:2}] Fix $\varepsilon>0$. Pick
  $V_0\in C_\Gamma$ and $\theta_0\in(\R^2)^*/(\Gamma')^*$. Then,
  $E_n(V_0,\theta_0)$ is an isolated eigenvalue of $H(\theta_0,V_0)$
  of multiplicity say $N_0=N(\theta_0,V_0)$. Let $\delta>0$ be such
  that $E_n(V_0,\theta_0)$ be the only eigenvalue of $H(\theta_0,V_0)$
  in $D(E_n(V_0,\theta_0),2\delta)$, the disk of center
  $E_n(V_0,\theta_0)$ and radius $2\delta$. The projector onto the
  eigenspace associated to $E_n(V_0,\theta_0)$ and $H(\theta_0,V_0)$
  is given by Riesz's formula
  \begin{equation*}
    \Pi(\theta_0,V_0)=\frac1{2i\pi}\int_{|z-E_n(V_0,\theta_0)|=\delta}
    (z-H(\theta_0,V_0))^{-1}dz.
  \end{equation*}
  It is well known (see e.g.~\cite{MR58:12429c,Ka:80}) that, there
  exists $\varepsilon_0\in(0,\varepsilon)$ such that, for
  $\|(\theta,V)-(\theta_0,V_0)\|<\varepsilon_0$, the projector onto
  the eigenspace associated to the spectrum of $H(\theta,V)$ in
  $D(E_n(V_0,\theta_0),\delta)$ is given by
  \begin{equation}
    \label{eq:19}
    \Pi(\theta,V)=\frac1{2i\pi}\int_{|z-E_n(V_0,\theta_0)|=\delta}
    (z-H(\theta,V))^{-1}dz.
  \end{equation}
  In particular, the rank of this projector is constant and equal to
  $N_0$, the multiplicity of $E_n(V_0,\theta_0)$ as an eigenvalue of
  $H(\theta_0,V_0)$.\\
  Consider the operator $M(\theta,V)=\Pi(\theta,V)H(\theta,V)$. Its
  eigenvalues are the eigenvalues of $H(\theta,V)$ in
  $D(E_n(\theta_0,V_0),\delta)$ and it has finite rank $N_0$. Let
  $(\psi_j)_{1\leq j\leq N_0}$ be an orthonormal basis of eigenvectors
  of $H(\theta_0,V_0)$ associated to the eigenvalue
  $E_n(\theta_0,V_0)$. For $j\in\{1,\cdots,N_0\}$, set
  $\psi_j(\theta,V)=\Pi(\theta,V)\psi_j$ and let $G(\theta,V)$ be the
  Gram matrix of these vectors. Then,
  \begin{equation*}
    G(\theta,V)-\text{Id}_{N_0}=O(\|(\theta,V)-(\theta_0,V_0)\|)
  \end{equation*}
  and the vectors
  \begin{equation*}
    \begin{pmatrix}
      \varphi_1(\theta,V)&\cdots&\varphi_{N_0}(\theta,V)
    \end{pmatrix}
    =\begin{pmatrix}\psi_1(\theta,V)&\cdots&\psi_{N_0}(\theta,V)
    \end{pmatrix}      \sqrt{G^{-1}(\theta,V)}
  \end{equation*}
  form an orthonormal basis of $\Pi(\theta,V)\mathcal{H}_{B,p}$.\\
  $E$ is an eigenvalue of $H(\theta,V)$ in
  $D(E_n(V_0,\theta_0),\delta)$ if and only if
  \begin{equation*}
    P(E;\theta,V)=\text{Det}\,(\tilde M(\theta,V)-E)=0 
  \end{equation*}
  where Det denotes the determinant and $\tilde M(\theta,V)$, the
  matrix of $M(\theta,V)$ in the basis
  $(\varphi_1(\theta,V),\cdots,\varphi_{N_0}(\theta,V))$.\\
  Then, either of two things occur:
  \begin{enumerate}
  \item there exists $\varepsilon>0$ and a function $(\theta,V)\mapsto
    E(\theta,V)$ such that, for
    $\|(\theta,V)-(\theta_0,V_0)\|<\varepsilon$, one has
    \begin{equation*}
      P(E(\theta,V),\theta,V)=\partial_EP(E(\theta,V),\theta,V)=\cdots
      =\partial^{N_0-1}_EP(E(\theta,V),\theta,V)=0
    \end{equation*}
    in which case, one has
    \begin{equation*}
      P(E,\theta,V)=(E-E(\theta,V))^{N_0}.
    \end{equation*}
    So $E_n(\theta,V)$ is the only eigenvalue of the matrix $\tilde
    M(\theta,V)$. For $\theta$ and $V$ real, $\tilde M(\theta,V)$ is
    Hermitian hence it is equal to $E_n(\theta,V)\,\text{Id}_{N_0}$.\\
    Pick now $V$ complex such that
    $\|(\theta,V)-(\theta_0,V_0)\|<\varepsilon/4$. We can write
    $V=V_r+iV_i$ with both $V_r\in\mathcal{C}_\Gamma$ and
    $V_i\in\mathcal{C}_\Gamma$. For $z\in D(0,2)$,
    $\|(\theta,V_r+zV_i)-(\theta_0,V_0)\|<\varepsilon$. Hence,
    $z\mapsto\tilde M(\theta,V_r+zV_i)$ and $z\mapsto
    E_n(\theta,V_r+zV_i)$ are analytic. Above, we have proved that,
    for $z$ real, one has
    \begin{equation*}
      \tilde M(\theta,V_r+zV_i)=E_n(\theta,V_r+zV_i)\text{Id}_{N_0}.
    \end{equation*}
    By analytic continuation, this stays true for $z$ in $D(0,2)$ in
    particular for $z=i$ i.e. $\tilde M(\theta,V)$ is Hermitian hence
    it is equal to $E_n(\theta,V)$ times the identity.\\
    So, $(\theta,V)\mapsto E(\theta,V)$ in an analytically degenerate
    eigenvalue of $H(\theta,V)$ (of order $N_0$).
    \begin{Rem}
      \label{rem:3}
      Actually, using the normal Jordan form for matrices instead of
      the Hermitian nature of the matrix, we only need to know that
      $\tilde M(\theta_0,V_0)$ is reducible to conclude that if the
      multiplicity of $E_n(\theta,V)$ is constant then the eigenvalue
      is analytically degenerate.
    \end{Rem}
  \item or, for any $\varepsilon>0$, there exists $N_1<N_0$ and
    $(\theta_1,V_1)$ such that
    $|(\theta_1,V_1)-(\theta_0,V_0)|<\varepsilon$ and
    \begin{equation*}
      \begin{split}
        P(E(\theta_1,V_1),\theta_1,V_1)&=\partial_EP(E(\theta_1,V_1),\theta_1,V_1)
        \\&=\cdots
        =\partial^{N_1-1}_EP(E(\theta_1,V_1),\theta_1,V_1)=0
      \end{split}
    \end{equation*}
    and
    \begin{equation*}
      \partial^{N_1}_EP(E(\theta_1,V_1),\theta_1,V_1)\not=0.;
    \end{equation*}
    in this case, $E(\theta_1,V_1)$ is an eigenvalue of multiplicity
    $N_1\leq N_0-1$ of $H(\theta_1,V_1)$.
  \end{enumerate}
  In the first case, Lemma~\ref{le:2} is proven. In the second case,
  we can then start the process over again near
  $(\theta_1,V_1)$. After at most $N_0$ such reductions we will have
  constructed the pair $(\theta_\varepsilon,V_\varepsilon)$ announced
  in Lemma~\ref{le:2}.  This completes the proof of Lemma~\ref{le:2}.
\end{proof}
We now turn to the proof of Lemma~\ref{le:1}.
\begin{proof}[Proof of Lemma~\ref{le:1}] Pick
  $\theta_0\in(\R^2)^*/(\Gamma')^*$ and $V_0\in C_\Gamma$. Assume that
  $E_n(\theta_0,V_0)$ is an analytically degenerate eigenvalue of
  $H(\theta_0,V_0)$. Let us write $E(\theta,V):=E_n(\theta,V)$. Assume
  that the conclusions of Lemma~\ref{le:1} is false. Then, there
  exists $\varepsilon>0$ such that for any $V\in\mathcal{C}_\Gamma$
  such that $\|V-V_0\|\leq\varepsilon$, the function $\theta\mapsto
  E(\theta,V)$ is constant. In particular, we can slightly change
  $V_0$ to assume that it is real analytic and the same conclusion
  still holds.  \\
  Pick $U\in \mathcal{C}_\Gamma$ such that $\|U\|=1$ and set
  $V_t=V_0+tU$, $t$ complex small.  As $E(\theta_0,V_0)$ is an
  analytically degenerate eigenvalue of $H(\theta_0,V_0)$, there
  exists $\varepsilon>0$ and $\varphi(\theta,t)$ real analytic in
  $(\theta,t)$ such that, for $|t|\leq\varepsilon$ and
  $|\theta-\theta_0|\leq\varepsilon$, one has
  \begin{equation}
    \label{eq:23}
    (H(\theta,t)-E(\theta,t))\varphi(\theta,t)=0,\quad \|\varphi(\theta,t)\|=1.
  \end{equation}
  Moreover $(\theta,t)\mapsto E(\theta,t)$ is real analytic.\\
  Differentiating the eigenvalue equation~\eqref{eq:23} for $\varphi$
  in $t$ yields
  \begin{equation}
    \label{eq:15}
    (H(\theta,t)-E(\theta,t))\partial_t\varphi(\theta,t)
    =[\partial_t E(\theta,t)-U]\varphi(\theta,t).
  \end{equation}  
  We note that, for $\gamma'\in\Gamma'$,
  \begin{equation*}
    W^B_{\gamma'}(\partial_t\varphi(\theta,t))=
    \partial_tW^B_{\gamma'}(\varphi(\theta,t))=
    \partial_t\varphi(\theta,t)
  \end{equation*}
  so $\partial_t\varphi(\theta,t)\in\mathcal{H}_{B,p}$. \\
  Using~\eqref{eq:23} and the self-adjointness of $H(\theta,t)$ on
  $\mathcal{H}_{B,p}$, one obtains
  \begin{equation}
    \label{eq:18}
    \partial_tE(\theta,t)=\langle U
    \varphi(\theta,t),\varphi(\theta,t)\rangle.
  \end{equation}
  We now assume that $E(\theta,t)$ does not depend on $\theta$ in some
  neighborhood of $\theta_0$ and for $t$ small i.e.
  \begin{equation*}
    \nabla_\theta E(\theta,t)=0.    
  \end{equation*}
  So differentiating~\eqref{eq:18} with respect to $\theta$, we obtain
  that
  \begin{equation*}
    \begin{split}
      0&=\partial_{t}\nabla_\theta E(\theta,t)=
      \nabla_\theta\partial_{t}E(\theta,t)\\
      &= \langle U\varphi(\theta,t),
      \nabla_\theta\varphi(\theta,t)\rangle+ \langle
      U\nabla_\theta\varphi(\theta,t),
      \varphi(\theta,t)\rangle\\
      &= 2\text{Re}\left[\langle U\varphi_k(\theta,t),
        \nabla_\theta\varphi_k(\theta,t)\rangle\right]
    \end{split}
  \end{equation*}
  Here, the real part is meant coordinate wise.\\
  At $t=0$, we then get that
  \begin{equation}
    \label{eq:13}
    \begin{split}
      0&=\text{Re}\left[\langle U\varphi(\theta,0),
        \nabla_\theta\varphi_k(\theta,0)\rangle\right]\\&=
      \int_{\R^2/\Gamma}U(x)\,\text{Re}\left(\nabla_\theta\varphi_k(x;\theta,0)
        \overline{\varphi_k(x;\theta,0)}\right)dx.
    \end{split}
  \end{equation}
  So, if for all $U\in\mathcal{C}_\Gamma$ such $\|U\|=1$ and for $t$
  small, we know that $\theta\mapsto E(\theta,t)$ is constant in some
  neighborhood of $\theta_0$, we obtain that~\eqref{eq:13} holds for
  all $U\in\mathcal{C}_\Gamma$ such that $\|U\|=1$. So, for $\theta$
  near $\theta_0$, one has
  \begin{equation}
    \label{eq:14}
    \forall x\in\R^2,\quad 2\text{Re}(\nabla_{\theta}\varphi(x;\theta,0)
    \overline{\varphi(x;\theta,0)})=
    \nabla_{\theta}\left(|\varphi(x;\theta,0)|^2\right)\equiv0.
  \end{equation}
  The operator $(i\nabla-A-\theta)^2+V_0$ being elliptic with real
  analytic coefficients, it is analytically hypoelliptic (see,
  e.g.~\cite{MR699623}); hence, $x\mapsto\varphi(x;\theta,0)$ is real
  analytic on $\R^2$. For $|\theta-\theta_0|\leq\varepsilon$, let
  $O_\theta\subset\R^2$ be the open set where the function
  $x\mapsto\varphi(x;\theta,0)$ does not vanish. By~\eqref{eq:14},
  this set is independent of $\theta$; we denote it by $O$. As
  $\varphi(\theta,0)\in\mathcal{H}_{B,p}$, $O$ is invariant by the
  translations by a vector in $\Gamma'$. Define $Z$ by
  $Z:=\R^2\setminus O$. $Z$ is also $\Gamma'$-periodic. Let $C$ be the
  fundamental cell of the lattice $\Gamma'$. As $Z$ is the set of
  zeros of the real analytic function $x\mapsto\varphi(x;\theta_0,0)$
  and as $Z\cap C$ is compact, we know that $Z\cap C$ has the
  following finite decomposition (see e.g.~\cite{MR89k:32011})
  \begin{equation}
    \label{eq:12}
    Z\cap C=\bigcup_{p=1}^{p_0}\mathcal{A}_p
  \end{equation}
  where the union is disjoint and, for $1\leq p\leq p_0$, one has
  \begin{enumerate}
  \item the set $\mathcal{A}_p$ either is reduced to a single point or
    is a connected real-analytic curve (i.e. a connected real analytic
    manifold of dimension 1);
  \item if $p\not=p'$ and $\mathcal{A}_p\cap
    \overline{\mathcal{A}_{p'}}\not=\emptyset$, then
    \begin{itemize}
    \item $\mathcal{A}_p\subset \overline{\mathcal{A}_{p'}}$,
    \item $\mathcal{A}_p$ is reduced to a single point,
    \item $\mathcal{A}_{p'}$ is a real-analytic curve;
    \end{itemize}
  \item assume $\mathcal{A}_p=\{x_0\}$ . Then, either $x_0$ is
    isolated in $Z\cap C$ or, for some $\varepsilon_0>0$ sufficiently
    small, one has
    \begin{equation*}
      Z\cap C\cap
      \Dot{\overline{D}}(x_0,\varepsilon_0)= \bigcup_{p'\in
        E}\mathcal{A}_{p'}\cap\Dot{\overline{D}}(x_0,\varepsilon_0)
    \end{equation*}
    where $E$ is a non empty, finite set of indices such that, for
    $p'\in E$, the set $\mathcal{A}_{p'}$ is a real analytic curve.\\
    Here, $\Dot{\overline{D}}(x_0,\varepsilon_0)=
    \{0<|x-x_0|\leq\varepsilon_0\}$.
  \end{enumerate}
  Let $Z_0=\cup_{\#\mathcal{A}_p=1}\mathcal{A}_p$ be the set of the
  points composing the point components in the above decomposition.
  \begin{Rem}
    \label{rem:4}
    As our Hamiltonian has no real symmetry i.e. the partial
    differential operator does not have real coefficients and as we
    are working in two space dimensions, it is reasonable to expect
    that the nodal set of an eigenfunction, if it is no empty, is
    actually made of points.
  \end{Rem}
  We will use the following
  \begin{Le}
    \label{le:4}
    Let $Z_\nabla$ be the set of points $x_0$ in $C$ such that
    $\varphi(x_0;\theta,0)=0$ and
    $\nabla\varphi(x_0;\theta,0)=0$. Then, $Z_\nabla$ consists of
    isolated points.
  \end{Le}
  \noindent We postpone the proof of Lemma~\ref{le:4} to complete that
  of Lemma~\ref{le:1}.\\
  Consider a horizontal straight line $L_x=x+\R\times\{0\}$ that does
  not intersect $Z_0\cup Z_\nabla$. As the other components of $Z$ are
  real analytic curves, possibly shifting this line, we can assume
  that it intersects these curves transversally in finitely many
  points. For $\delta>0$, define the strip $S^\delta_x$ by
  \begin{equation*}
    S^\delta_x=x+\R\times(-\delta,\delta).
  \end{equation*}
  Then, there exists $\delta>0$ such that
  \begin{itemize}
  \item $\overline{S^\delta_x}\cap (Z_0\cup Z_\nabla)=\emptyset$,
  \item $S^\delta_x$ intersects $Z$ in $C$ at, at most, finitely many
    vertical curves, and these curves partition the strip in a finite
    number of open domains (see Fig.~\ref{fig:1}).  Here, vertical
    means that the curves can be parametrized by the coordinate
    $x_1$.
  \end{itemize}
  \begin{figure}[h]
    \centering
    \includegraphics[height=5cm]{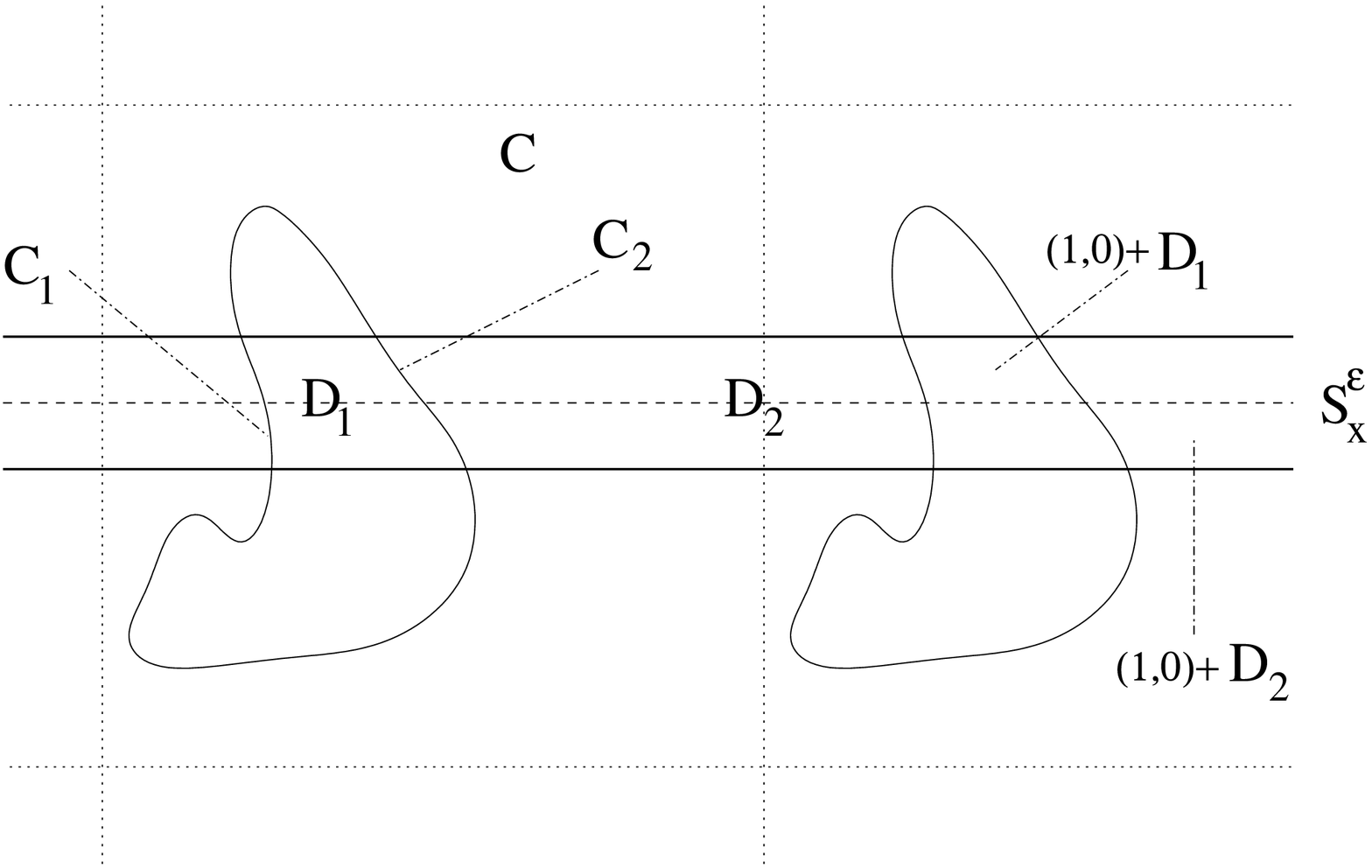}
    \caption{The strip}
    \label{fig:1}
  \end{figure}
  Recall that $Z$ is $\Gamma'$ periodic. Hence, we get that
  \begin{equation*}
    S^\delta_x\setminus Z=\bigcup_{\gamma'\in q\Z e1}
    \bigcup_{k=1}^s\gamma'+D_k\quad\text{ and }\quad
    Z\cap S_x=\bigcup_{\gamma'\in q\Z e1}\bigcup_{k=1}^s\gamma'+C_k
  \end{equation*}
  where, to fix ideas we assume that $C_k$ is the left boundary of
  $D_k$.\\
  We prove
  \begin{Le}
    \label{le:3}
    Let $D$ be one of the domains $\gamma'+D_k$ for some $1\leq k \leq
    s$ and some $\gamma'\in q\Z e_1$. \\
    For $\theta$ such that $|\theta-\theta_0|<\varepsilon$, there
    exists two real continuous $x\in\overline{D}\mapsto
    g_D(x;\theta)\in\R$ and $x\in\overline{D}\mapsto\psi_D(x)\in\R^+$,
    such that
    \begin{equation}
      \label{eq:17}
      \forall x\in D,\quad \varphi(x;\theta,0)=e^{i g_D(x;\theta)}\psi_D(x).
    \end{equation}
    and such that
    \begin{itemize}
    \item for any $x_0\in D$, $(x,\theta')\mapsto g_D(x;\theta')$ is
      real analytic in a neighborhood of $(x_0,\theta)$,
    \item let $D'$ be another domain in the collection
      $(\gamma'+D_k)_{\gamma',k}$; if
      $\overline{D}\cap\overline{D'}\not=\emptyset$ and $D'$ is to the
      left of $D$, then, for $x\in\overline{D}\cap\overline{D'}$, one
      has
      \begin{equation}
        \label{eq:16}
        g_D(x;\theta)=g_{D'}(x;\theta)+\pi .
      \end{equation}
    \end{itemize}
  \end{Le}
  \noindent Before turning to the proof of this result, let us
  complete the proof of Lemma~\ref{le:1}.\\
  Recall that $\varphi(\theta,0)\in\mathcal{H}_{B,p}$ i.e. that
  $W_{q,\gamma'}^B\varphi(\theta,0)=\varphi(\theta,0)$ for all
  $\gamma'\in\Gamma'$. By~\eqref{eq:8}, the definition of
  $W_{q,\gamma'}^B$, the functions coming into the decomposition given
  in Lemma~\ref{le:3} must satisfy, for $\gamma'\in q\Z e_1$ and $x\in
  D_k$
  \begin{equation}
    \label{eq:20}
    g_{\gamma'+D_k}(x+\gamma',\theta)=g_{D_k}(x,\theta)-\frac{B}2x\wedge\gamma'
    -\pi\gamma'_1\gamma'_2
  \end{equation}
  and
  \begin{equation*}
    \psi_{\gamma'+D_k}(x+\gamma')=\psi_{D_k}(x).
  \end{equation*}
  For $D$, one of the domains $(\gamma'+D_k)_{\gamma',k}$, plug the
  representation~\eqref{eq:17} into the eigenvalue
  equation~\eqref{eq:23} to obtain that, on $D$, one has
  \begin{equation*}
    (i\nabla_x-A-\theta-\nabla_x g_D)^2\psi_D+V_0\psi_D=E\psi_D
  \end{equation*}
  where $E=E(\theta,0)$ as it does not depend on $\theta$. As $V_0$,
  $\psi$ and $g$ real valued, we can take the complex conjugate of
  this equation to obtain that, on $O$, one has
  \begin{equation*}
    (i\nabla_x+A+\theta+\nabla_x g_D)^2\psi_D+V_0\psi_D=E\psi_D.
  \end{equation*}
  Summing the last two equations, one finally obtains that, on $D$,
  one has
  \begin{equation*}
    (A+\theta+\nabla_x g_D)^2\psi_D=(E-V_0)\psi_D+\Delta\psi_D
  \end{equation*}
  As $A$, $\psi_D$, $E$ and $V_0$ do not depend on $\theta$, as
  $\psi_D$ does not vanish on $D$, this equation implies that, for
  $x\in D$, the function $\theta\mapsto \theta+\nabla_xg_D(x,\theta)$
  does not depend on $\theta$. Hence, there exists a function
  $x\in\overline{D}\mapsto h_D(x)$ that is real analytic in $D$ and a
  real analytic function $\theta\mapsto c_D(\theta)$ such that, for
  $|\theta-\theta_0|\leq\varepsilon$ and $x\in D$, one has
  \begin{equation}
    \label{eq:24}
    g_D(x,\theta)=-\theta\cdot x+h_D(x)+c_D(\theta).
  \end{equation}
  We note that~(\ref{eq:16}) in Lemma~\ref{le:3} tells us that, if
  $D'$ is to the left of $D$ and
  $\overline{D'}\cap\overline{D}\not=\emptyset$, then we may choose
  \begin{equation}
    \label{eq:22}
    c_D(\theta)=c_{D'}(\theta)+\pi.
  \end{equation}
  We now plug the representation~(\ref{eq:24}) into~\eqref{eq:20} and
  use~(\ref{eq:22}) to obtain that, for
  $|\theta-\theta_0|\leq\varepsilon$, $\gamma'=\gamma'_1 e_1\in p\Z
  e_1$ and $x\in D$,
  \begin{equation*}
    \theta\cdot\gamma'=h_{\gamma'+D}(x)-h_D(x)
    +\frac{B}2x\wedge\gamma'+\pi\gamma'_1\gamma'_2-s\gamma'_1\pi.    
  \end{equation*}
  This is absurd as the left hand side of this expression depends on
  $\theta$ and the right does not.\\
  This completes the proof of Lemma~\ref{le:1}.
\end{proof}
\noindent We now turn to the proof of Lemmas~\ref{le:4}
and~\ref{le:3}.
\begin{proof}[Proof of Lemma~\ref{le:4}]
  First, the set $Z_\nabla\cap C$ is real analytic so can be
  decomposed in the same way as $Z\cap C$. If it does not consist of
  isolated points, then it contains an analytic curve, say, $c$. Pick
  a point $x^0$ in this curve. Near $x^0=(x_1^0,x_2^0)$ assume,
  without restriction, that the curve is parametrized by $x_2=c(x_1)$
  where $c$ is real analytic.\\
  Define the functions $u(x)=\text{Re}(\varphi(x;\theta,0))$ and
  $v(x)=\text{Im}(\varphi(x;\theta,0))$. They are real analytic, real
  valued and satisfy
  \begin{itemize}
  \item as $\varphi(\theta,0)$ is a solution to the eigenvalue
    equation~(\ref{eq:23}),
    \begin{equation}
      \label{eq:21}
      \begin{aligned}
      (-\Delta u)+(A-\theta)^2u+2A\cdot\nabla v=(E-V)u,\\
      (-\Delta v)+(A-\theta)^2v-2A\cdot\nabla u=(E-V)v;        
      \end{aligned}
    \end{equation}
    here, we used div$A=0$;
  \item on $c$, one has
    \begin{equation}
      \label{eq:25}
      0=u=v=\partial_1u=\partial_1v =\partial_2u=\partial_2v
    \end{equation}
    by the definition of $Z_\nabla$.
  \end{itemize}
  Let us prove inductively that, for any $\alpha\in\N^2$,
  $\partial^\alpha u=\partial^\alpha v=0$ on $c$. Assume that, for
  $\alpha_1+\alpha_2\leq N$, one has
  $\partial_1^{\alpha_1}\partial_2^{\alpha_2}u=
  \partial_1^{\alpha_1}\partial_2^{\alpha_2}v=0$. Let us prove that it
  still holds for $\alpha_1+\alpha_2= N+1$.\\
  For $\alpha_1+\alpha_2\leq N+1$, differentiating $\alpha_1-1$ times
  equations~(\ref{eq:21}) in $x_1$ and $\alpha_2-1$ times in $x_2$
  yields that, on $c$, one has
  \begin{equation}
    \label{eq:27}
      \begin{aligned}
      \partial_1^{\alpha_1+1}\partial_2^{\alpha_2-1}u+
      \partial_1^{\alpha_1-1}\partial_2^{\alpha_2+1}u
      =\sum_{\beta_1+\beta_2\leq N}a_{\beta_1\beta_2}
      \partial^\beta u+b_{\beta_1\beta_2}\partial^\beta v=0,\\
      \partial_1^{\alpha_1+1}\partial_2^{\alpha_2-1}v+
      \partial_1^{\alpha_1-1}\partial_2^{\alpha_2+1}v
      =\sum_{\beta_1+\beta_2\leq N}c_{\beta_1\beta_2}
      \partial^\beta u+d_{\beta_1\beta_2}\partial^\beta v=0.
      \end{aligned}
  \end{equation}
  Differentiating $\partial_1^{\alpha_1}\partial_2^{\alpha_2}u=0$
  along $c$, we get
  \begin{equation}
    \label{eq:26}
    \left(\partial_1^{\alpha_1+1}\partial_2^{\alpha_2}u\right)(x_1,c(x_1))+
    c'(x_1)\left(\partial_1^{\alpha_1}\partial_2^{\alpha_2+1}u\right)(x_1,c(x_1))=0
  \end{equation}
  Using this for $(\alpha_1,\alpha_2)=(N,0)$ and
  $(\alpha_1,\alpha_2)=(N-1,1)$ and the first equation
  in~(\ref{eq:27}) for $(\alpha_1,\alpha_2)=(N,1)$, we get the system
  \begin{equation*}
    \begin{cases}
      \partial_1^{N+1}u+c'\partial_1^N\partial_2u&=0\\
      \partial_1^{N}\partial_2u+c'\partial_1^{N-1}\partial_2^2u&=0\\
      \partial_1^{N+1}u+c'\partial_1^{N-1}\partial^2_2u&=0
    \end{cases}
  \end{equation*}
  which implies that
  \begin{equation*}
    \partial_1^{N+1}u=\partial_1^N\partial_2u=\partial_1^{N-1}\partial^2_2u=0.
  \end{equation*}
  Let us assume that $c'(x_1)\not=0$. Then, using~(\ref{eq:27})
  inductively, we get that
  $\partial_1^{N+1-\alpha}\partial_2^{\alpha}u=0$ for all $0\leq
  \alpha\leq N+1$.\\
  If $c'$ does not vanish on the whole curve, we just work near a
  point where it does not vanish. If $c'$ vanishes on the whole curve,
  then the curve is a straight horizontal line, say, $x_2=0$ and we
  proceed as follows. By differentiation of~(\ref{eq:25}), we
  immediately get that, on $c$, one has
  \begin{equation*}
   \partial_1^{N+1}u=\partial_1^N\partial_2u=
   \partial_1^{N+1}v=\partial_1^N\partial_2v=0 
  \end{equation*}
  Then,~(\ref{eq:27}) and the induction assumption yield, for $0\leq
  \alpha\leq N$,
  \begin{gather*}
    0=-\partial_1^{N+1-\alpha}\partial_2^\alpha u=
    \partial_1^{N-\alpha-1}\partial_2^{\alpha+2}u,\\
    0=-\partial_1^{N+1-\alpha}\partial_2^\alpha v=
    \partial_1^{N-\alpha-1}\partial_2^{\alpha+2}v.
  \end{gather*}
  Finally we proved that, if $Z_\nabla\cap C$ contains a curve, the
  functions $(\partial^{\alpha}_x)\varphi(\theta,0)$ vanish
  identically on this curve. As $\varphi(\theta,0)$ is real analytic,
  this implies that this function vanishes identically which
  contradicts the assumption that its norm in $\mathcal{H}_{B,p}$ is
  $1$.\\
  This completes the proof of Lemma~\ref{le:4}.
\end{proof}
\begin{proof}[Proof of Lemma~\ref{le:3}] Clearly, in the domains
  $(D_k)_{\leq k\leq s}$ and their translates, the
  decomposition~(\ref{eq:17}) is the decomposition into argument and
  modulus of the complex number $\varphi(x;\theta,0)$. As
  $\varphi(x;\theta,0)$ does not vanish, its argument and modulus are
  also real analytic. So we only need to study what happens at the
  crossing of one of the curves $(C_k)_{\leq k\leq s}$. So, we study
  $x\mapsto\varphi(x;\theta,0)$ near $x^0\in C_k$.\\
  As $S^\delta_x\cap (Z_0\cup Z_\nabla)=\emptyset$, we know that
  $\nabla\varphi(x^0,\theta,0)\not=0$. Using the notation of the proof
  of Lemma~\ref{le:4} i.e.  $u(x)=\text{Re}(\varphi(x;\theta,0)$ and
  $v(x)=\text{Im}(\varphi(x;\theta,0)$, we may assume that $\nabla
  u(x^0)\not=0$. As the curve $C_k$ is vertical, we know that
  $\partial_1u(x^0)\not=0$. We can then find a real analytic change of
  variables that maps a neighborhood of $x^0$ into a neighborhood of
  $0$ and that maps the set $\{x;\ u(x)=0\}$ into the straight line
  $\{x_1=0\}$. We perform this change of variables on $u$ and $v$ and
  call the function thus obtained again $u$ and $v$. Then, in a
  neighborhood of $0$, one has that
  \begin{equation}
    \label{eq:28}
    u(x_1,x_2)=0\Leftrightarrow
    x_1=0,\quad \partial_1u(0,0)\not=0,\quad v(0,x_2)=0.
  \end{equation}
  The functions $u$ and $v$ being real analytic, we can write them as
  \begin{equation*}
    u(x_1,x_2)=\tilde w(x_2)+x_1w(x_1,x_2)\quad\text{ and}\quad 
    v(x_1,x_2)=\tilde t(x_2)+x_1t(x_1,x_2)
  \end{equation*}
  where all the functions are real analytic.\\
  Then,~(\ref{eq:28}) implies then that
  \begin{equation*}
    w(0,0)\not=0,\quad \tilde w(x_2)=\tilde t(x_2)=0\text{ identically}.   
  \end{equation*}
  Hence, we obtain that
  \begin{equation*}
    (u+iv)(x_1,x_2)=x_1(w+it)(x_1,x_2)\text{ where }|(w+it)(0,0)|\not=0.
  \end{equation*}
  Changing back to the initial variables, if $x_1\mapsto c(x_1)$ is a
  parametrization of the curve $C_k$ near $x^0$, we see that, in $U$,
  a neighborhood of $x^0$, we can write
  \begin{equation*}
    \varphi(x;\theta,0)=(x_2-c(x_1))\psi(x)\text{ where }\psi(x^0)\not=0.
  \end{equation*}
  Hence, for $x\in D_k\cap U$, one has
  \begin{equation*}
    e^{i g_{D_k}(x;\theta)}\psi_{D_k}(x)=(x_2-c(x_1))\psi(x),\quad
    x_2\geq c(x_1)
  \end{equation*}
  and for $x\in D_{k-1}\cap U$, one has
  \begin{equation*}
    e^{i g_{D_{k-1}}(x;\theta)}\psi_{D_{k-1}}(x)=(x_2-c(x_1))\psi(x)=
    -(c(x_1)-x_2)\psi(x),\quad  x_2\leq c(x_1) .
  \end{equation*}
  This implies that we can continue $g_{D_{k-1}}$ and $g_{D_k}$
  continuously up to the boundary $C_k$ and that they satisfy the
  relation~(\ref{eq:16}) on $C_k$.\\
  This completes the proof of Lemma~\ref{le:3}.
\end{proof}
%
%
\def\cprime{$'$} \def\cydot{\leavevmode\raise.4ex\hbox{.}}

\end{document}